\documentclass[12pt,a4paper]{article}
\usepackage{upgreek}
\pdfoutput=1
\usepackage{endnotes}
\usepackage{enumerate}
\usepackage{amsmath,amssymb,graphicx,alltt,xy,amsthm}
\usepackage{dsfont}

\usepackage{comment}
\usepackage{bbm}
\usepackage{subcaption}
\usepackage{float}

\usepackage{overpic}
\date{}
\input xy
\xyoption{all}
\newtheorem{theorem}{Theorem}
\newtheorem{lemma}{Lemma}

\newcommand*{\rom}[1]{\expandafter\@slowromancap\romannumeral #1@}

\usepackage{tikz}
\usetikzlibrary{shapes, backgrounds}
\usetikzlibrary{shapes}

\title{Explicit solutions of the Bass and SI models on hypernetworks}
\author{Gadi Fibich$^1$, Juan G. Restrepo$^2$, and Guy Rothmann$^1$}

\begin{document}
	\maketitle
	\vspace{-3em} 
	\footnotesize
	\begin{flushleft}
		$^1$ Department of Applied Mathematics, Tel Aviv University, fibich@tau.ac.il, guy86222@gmail.com\\
		$^2$ Department of Applied Mathematics, University of Colorado at Boulder, juanga@colorado.edu
	\end{flushleft}
\centerline{
	\small
	\date{\today}
}
\small

\begin{abstract}
	We analyze the Bass model and the Susceptible-Infected (SI)  model on hypergraphs with 3-body interactions. We derive the master equations for general hypernetworks, and use them to obtain explicit expressions for the expected adoption/infection level on infinite complete hypernetworks, infinite Erd\H{o}s-Rényi hypernetworks, and on infinite hyperlines. These expressions are exact, as they are derived without making any approximation.
\end{abstract}

	
\section{Introduction}
Spreading processes on networks have been studied in various research areas, including mathematics, physics, marketing, computer science, and sociology~\cite{barabasi1999emergence}. 
In marketing, the study of the diffusion of innovations began with 
the seminal work of Bass in 1969~\cite{bass1969new}, which inspired a huge body of theoretical and empirical research~\cite{Mahajan-93}. In epidemiology, 
mathematical models have been used to study the spread of infectious diseases in social networks~\cite{anderson1992infectious,keeling2005networks,kiss2017mathematics}.
A key question in these studies has been the role that the network structure plays in the spreading process. 

In the Bass and SI models on networks, the overall rate of peer influences on a susceptible individual 
is the {\em sum} of the influence rates exerted by their peers which are adopters/infected. This assumption is a reasonable starting point. In some cases, however, it is more realistic to use a {\em threshold model} in which the decision to adopt the product  takes place only if the number of adopters exceeds a certain threshold at which the net benefits for adopting the product begin to exceed the net costs~\cite{Granovetter-78}. In other cases, the marginal influence of an adopter may be a decreasing function of the number of adopters who  have already influenced the nonadopter.

In order to allow for a nonlinear dependence of the overall rate of peer influence
on the individual peer influences, it is natural to model these processes on
hypergraphs~\cite{berge1973graphs}.
Indeed, in recent years, extensive research has been devoted to spreading processes on hypergraphs.
For example, Palafox-Castillo et al.\ \cite{PALAFOXCASTILLO2022128053} used the mean-field approach to analyze the steady state of the SIR model on simplicial complexes.
Matamalas et al.\ \cite{Matamalas} used a microscopic Markov chain approximation to find abrupt phase transitions in the SIS model on simplicial complexes.
Kim et al.\ \cite{KimPhysRevE} used the facet approximation on random nested hypergraphs to compute the steady state in the SIS model, and showed that the hyperedge-nestedness affects the phase diagram significantly. 
\mbox{Arenas et al.\ \cite{Arenas}} used the triadic approximation in the SIS model on hypergraphs to demonstrate the double-edged effect of increased overlap between two- and three-body interactions: it decreases the invasion threshold, but also results in generally smaller outbreaks.
Iacopini et al.\ \cite{Iacopini2019} used the mean-field approximation in the SIS model on simplicial complexes and found discontinuous transitions in the steady state, and the emergence of bistable regions where both healthy and endemic states coexists. 
Higham and de Kergorlay~\cite{Higham2021} used the mean-field approximation in the SIS  model on hypergraphs to obtain spectral conditions for the local asymptotic stability of the zero-infection state. 
Arruda et al.\ \cite{deArruda2020} used the assumption that the variables describing the infected state of the nodes are independent in the SIS model on hypergraphs to show that the model exhibits a vast parameter space, including first- and second-order transitions, bistability, and hysteresis.  
Bianconi \cite{Bianconi2021} used the mean-field approximation to derive macroscopic equations in the SIS model on simplicial complexes, and found discontinuous transitions and bistability regions.

In this paper, we analyze the Bass and SI models on hypergraphs.
To simplify the presentation, we only consider hypergraphs with pure 3-body interactions.
The extension of our results and methodology to more general hypergraphs is straightforward.  
 We first derive the master equations for general hypergraphs with pure 3-body interactions.
 We then solve these equations explicitly, and obtain explicit expressions for the expected adoption/infection level as a function of time, for infinite complete hypernetworks, infinite Erd\H{o}s-Rényi hypernetworks, and infinite hyperlines. These expressions are exact, as they are derived without making any approximation. In all cases, we present careful numerical simulations that confirm the validity of the explicit expressions.

     Our work differs from previous studies in several aspects: 
\begin{itemize}
	\item We obtain explicit expressions for the expected adoption/infection level as a function of time,
	 whereas previous studies focused more on the limiting steady-state. 
	 These expressions allow us to address questions such as, e.g., the time for half of the population to become infected.
	  
	 \item The availability of explicit solutions simplifies the analysis, since it is much easier to analyze an explicit solution than to analyze the original stochastic network model.
	 
	\item  From a methodological persective, we start from the full system of the master equations, and solve it exactly, without applying any approximate closure at the level of pairs or triplets. 
	As a result, the explicit solutions that we obtain are exact, and not approximate.
	
	\item   To the best of our knowledge,  this is the first study of the Bass model on hypergraphs. 
	\item   We analyze the SI model and the Bass model using a unified framework.
\end{itemize}

%
%

The paper is organized as follows:
In Section~\ref{Bass_SI_general} we define the Bass and SI models in 3-body hypernetworks. In Section~\ref{master_eqs} we derive the master equations that describe the dynamics of diffusion/infection across hypernetworks. In Section~\ref{complete_hypernetworks} we derive explicit solutions to the master equations for complete 3-body hypernetworks. 
In Section~\ref{initial_dynamics} we explore the initial dynamics of the expected adoption level in the Bass model on arbitrary 3-body hypernetworks, and prove that the adoption rate initially decreases, regardless of the hypernetwork structure or parameters. This is the only case where we observe a qualitative difference between the 
spreading dynamics on graphs and on hypergraphs. 
In Section~\ref{ER_hypernetworks} we examine the spreading dynamics on Erdős-Rényi hypernetworks. In Section~\ref{infinite_hyperlines} we derive explicit solutions to the master equations for infinite 3-body hyperlines, supplemented by numerical simulations that validate the theoretical results.
Section~\ref{sec:final} concludes with some final remarks and suggestions for extensions of this study.

\section{Bass/SI~model on 3-body hypernetworks}
\label{Bass_SI_general}

The Bass model describes the adoption of new products or innovations within a population. In this framework, all individuals start as non-adopters and can transition to becoming adopters due to two types of influences: external factors, such as exposure to mass media, and internal factors where individuals are influenced by their peers who have already adopted the product.
The SI~model is used to study the spreading of infectious diseases within a population. In this model, some individuals are initially infected (the ``patient zero'' cases), and all subsequent infections occur through internal influences, whereby infected individuals transmit the disease to their susceptible peers, and infected individuals remain contagious indefinitely.
In both models, once an individual becomes an adopter/infected, it remains so at all later times.  In particular, she or he remain ``contagious'' forever. It is convenient to unify these two models into a single model,
the Bass/SI~model on networks, as follows. 
The difference between the SI~model and the Bass model is the lack of external influences in the former, and the lack of ``adopters zero'' in the latter.

Consider $M$~individuals, denoted by ${\cal M}:=\{1, \dots, M\}$.
We denote by $X_j(t)$ the state of individual~$j$ at time~$t$, so that
\begin{equation*}
	X_j(t)=\begin{cases}
		1, \qquad {\rm if}\ j\ {\rm is \ adopter/infected \ at\ time}\ t,\\
		0, \qquad {\rm otherwise,}	
	\end{cases}
	\qquad j \in \cal M.
\end{equation*}
The initial conditions at $t=0$ are stochastic, so that
\begin{subequations}
	\label{Bass_SI_hypernetworks_general}
	\begin{equation}
		\label{eq:general_initial}
		X_j(0)=  X_j^0 \in \{0,1\}, \qquad j\in {\cal M},
	\end{equation}
	where
	\begin{equation}
		\mathbb{P}(X_j^0=1) =I_j^0, \quad
		\mathbb{P}(X_j^0=0) =1-I_j^0,\quad I_j^0 \in [0, 1],  \qquad
		j \in \cal M,
	\end{equation}
	and
	\begin{equation}
		\label{eq:p:initial_cond_uncor-two_sided_line}
		\mbox{the random variables $\{X_j^0 \}_{j \in \cal M}$ are independent}.
	\end{equation}
	{\em Deterministic initial conditions} are a special case where
	$I_j^0 \in \{0,1\}$.	
	As long as $j$ is a nonadopter/susceptible, its adoption/infection rate at time~$t$ is~\footnote{\hbox{For comparison, the adoption/infection rate on two-body networks is~$\lambda_j=p_j+\sum_{k=1}^M q_{k,j}X_k(t).$}}
	\begin{equation}
	\label{eq:lambda_j_3body(t)}
	\lambda_j(t) = p_j+\sum\limits_{k_1,k_2=1}^M q_{k_1,k_2,j} X_{k_1}(t) X_{k_2}(t) ,
	\qquad j \in {\cal M}.
	\end{equation}
	Here, $p_j\geq 0$ is the rate of external influences on~$j$, and~$q_{k_1,k_2,j}\geq 0$ is the rate of internal influences by~$k_1$ and~$k_2$ on~$j$, provided that $k_1$ and $k_2$ are already adopters/infected. In addition, $q_{k_1,k_2,j}>0$ if $k_1,k_2$ and $j$ are distinct and the directional hyperdge $\{k_1,k_2\}\rightarrow j$ exists. Otherwise, $q_{k_1,k_2,j}=0$. Once~$j$ becomes an adopter/infected, it remains so at all later times.\footnote{i.e., the only admissible transition is
		$X_j=0 \to X_j=1$.}
	Hence, as $ \Delta t \to 0$,
	\begin{equation}
		\label{eq:general_model}
		\mathbb{P} (X_j(t+\Delta  t )=1  \mid   {\bf X}(t))=
		\begin{cases}
			\lambda_j(t) \, \Delta t , &  {\rm if}\ X_j(t)=0,
			\\
			1,\hfill & {\rm if}\ X_j(t)=1,
		\end{cases}
		\qquad                 j \in {\cal M},
	\end{equation}
	where ${\bf X}(t) := \{X_j(t)\}_{j \in \cal M}$ is the state of the network at time~$t$,	and
	\begin{equation}
		\label{eq:Bass-SI-models-ME-independent}
		\mbox{the random variables $\{X_j(t+\Delta  t )  \mid   {\bf X}(t) \}_{j \in \cal M}$ are independent}.
	\end{equation} 

\end{subequations}

In the Bass model there are no adopters when the product is first introduced into the market, and so $I_j^0 \equiv 0$ and $p_j>0$ for $j\in\mathcal{M}$. In the SI~model there are only internal influences for $t>0$, and so $p_j \equiv 0$ and $I^0_j>0$ for $j\in\mathcal{M}$.

The quantity of most interest is the expected  adoption/infection level
\begin{equation}
	\label{eq:number_to_fraction-general}
	f(t):=	\frac{1}{M} \sum_{j=1}^{M} f_j(t),
\end{equation}
where $f_j :=\mathbb{E}[X_j]$ is the adoption/infection probability of node~$j$.

\section{Master equations}
\label{master_eqs}
The most important analytic tool for the Bass/SI model on networks are the master equations, which were derived in~\cite{fibich2022diffusion}. In this section, we derive the master equations for general hypernetworks with 3-body interactions. 
Let $\Omega \subset  {\cal M}$ be a nontrivial subset of the nodes,  let
$\Omega^{\rm c}:={\cal M} \setminus \Omega$,  and let 
$$
S_{\Omega}(t):=\{X_{m}(t)=0,~m \in \Omega \},
\qquad  [S_{\Omega}](t)
:= \mathbb{P}(S_{\Omega}(t)),
$$
denote the event that all the nodes in~$\Omega$ are nonadopters/susceptibles at time~$t$, and
the probability of this event, respectively.
In what follows, we will use the notations
\begin{equation}
	\label{Set_notation}
	[S_{\Omega,k_1}]:=[S_{\Omega\cup \{k_1\}}],\quad [S_{\Omega,k_1,k_2}]:=[S_{\Omega\cup \{k_1,k_2\}}], \quad p_\Omega:=\sum_{m\in \Omega}p_m, \quad q_{k_1,k_2,\Omega}:= \sum_{m \in \Omega} q_{k_1,k_2,m},
\end{equation}
where $k_1,k_2 \in \Omega ^c$.
	\begin{theorem}
		\label{thm:master-eqs-general-hyper}
		The master equations for the Bass/SI  model~\eqref{Bass_SI_hypernetworks_general} on 3-body hypernetworks are
		\begin{subequations}
			\label{eqs:master-eqs-general-hypergraph}
			\begin{equation}
				\label{eq:master-eqs-general-hypergraph}
					\frac{d[S_{\Omega}]}{dt}= 
					-\left(p_{\Omega}+\sum_{k_1,k_2\in \Omega^{\rm c}} q_{k_1,k_2,\Omega}\right)[S_{\Omega}]+\sum_{k_1,k_2\in \Omega^{\rm c}}q_{k_1,k_2,\Omega}
					\biggl([S_{\Omega,k_1}]+[S_{\Omega,k_2}]-[S_{\Omega,k_1,k_2}]\biggr), 
			\end{equation}
			subject to the initial conditions
			\begin{equation}
				\label{eq:master-eqs-genera-icl-hypergraph}
				[S_{\Omega}](0)=\prod_{m\in \Omega}\left(1-I^0_m\right), 
			\end{equation}
			for all $\emptyset\not=\Omega \subset {\cal M}$.
		\end{subequations}
	\end{theorem}
\begin{proof}
	Consider the average over an infinite number of realizations 
of the Bass/SI  model~\eqref{Bass_SI_hypernetworks_general}. 
By definition, the event $S_\Omega(t)$ occurs at
a fraction~$[S_\Omega](t)$ of these realizations.
Since the only allowed transition is $S \to I$, 
new $S_{\Omega}$ realizations cannot be created,
and the existing $S_{\Omega}$~realizations are destroyed 
whenever any of the nodes in~$\Omega$ adopts. 
The adoption rate of node~$m \in \Omega$ is 
$\lambda_{m}(t)$, see~\eqref{eq:lambda_j_3body(t)}. Thus,
\begin{itemize}
	\item An existing $S_{\Omega}$~realization is destroyed if node~$m$ adopts externally. 
	Since there are $[S_{\Omega}]$ such realizations, this external influence leads to a reduction in~$[S_{\Omega}]$ at the rate of $p_m [S_{\Omega}]$.
	
	\item  An existing $S_{\Omega}$~realization is also destroyed if node~$m$ adopts as a result of an internal influence by some nodes~$k_1,k_2\in \Omega^{\rm c}$.
	For this to occur, at time $t$ all nodes of~${\Omega}$ should be nonadopters, and nodes~$k_1,k_2\in  \Omega^{\rm c}$ should be adopters. Denote this event by 
	$S_{\Omega}(t)\cap I_{k_1,k_2}(t)$, and the probability of this event
	by~$[S_{\Omega} \cap I_k](t)$. 
	Since there are $[S_{\Omega} \cap I_{k_1,k_2}](t)$ such realizations, 
	this external influence leads to a reduction  in~$[S_{\Omega}]$ at the rate of~$ q_{k_1,k_2,m} [S_{\Omega} \cap  I_{k_1,k_2}]$.
\end{itemize}
Therefore, the rate of change in  $[S_{\Omega}]$ due to external or internal influences on node~$m$ is
$$
- p_m [S_{\Omega}]
- \sum_{k_1,k_2\in \Omega^{\rm c}} q_{k_1,k_2,m} [S_{\Omega} \cap I_{k_1,k_2}] .
$$
The overall rate of change in  $[S_{\Omega}]$  is the sum of 
the rates of change in~$[S_{\Omega}]$ due to the adoptions of all nodes~$m \in \Omega$.
Hence,

\begin{equation}
	\label{eq:d_dt_[S_Omega]-inter}
	\frac{d[S_{\Omega}]}{dt} = -p_{\Omega}[S_{\Omega}]
- \sum_{k_1,k_2\in \Omega^{\rm c}} q_{k_1,k_2,\Omega} [S_{\Omega} \cap I_{k_1,k_2}].
\end{equation}

In order to express the master equations using only nonadoption probabilities, we first write $S_{\Omega}$ as the union of  four disjoint sets, 
$$
S_{\Omega} = S_{\Omega,k_1,k_2} 
\cup \left(S_{\Omega,k_1} \cap I_{k_2}\right)
\cup \left(S_{\Omega,k_2} \cap I_{k_1}\right)
\cup \left(S_{\Omega} \cap I_{k_1,k_2}\right).
$$
Therefore, 
$$
[S_{\Omega}] = [S_{\Omega ,k_1,k_2}] 
+ [S_{\Omega,k_1} \cap I_{k_2}]
+[S_{\Omega,k_2} \cap I_{k_1}]
+[S_{\Omega} \cap I_{k_1,k_2}].
$$
In addition, $S_{\Omega,k_1}$ can be written as the union of  two disjoint sets,  
$$
S_{\Omega,k_1} = S_{\Omega,k_1,k_2} \cup 
\left(S_{\Omega,k_1} \cap I_{k_2}\right).
$$
Hence, 
$$
 [S_{\Omega,k_1}]  =  [S_{\Omega,k_1,k_2}]+[S_{\Omega,k_1} \cap I_{k_2}].
$$
Similarly, 
$$
[S_{\Omega,k_2}]  =  [S_{\Omega,k_1,k_2}]+[S_{\Omega,k_2} \cap I_{k_1}].
$$
Combining the above, we have that 
\begin{equation}
	\label{eq:[S_Omega_cap_I_k,n]}
	[S_{\Omega} \cap I_{k_1,k_2}] = [S_{\Omega}] -[S_{\Omega,k_1}]-[S_{\Omega,k_2}]+[S_{\Omega,k_1,k_2}]. 
\end{equation}
Equation~\eqref{eq:master-eqs-general-hypergraph} follows from~\eqref{eq:d_dt_[S_Omega]-inter} and~\eqref{eq:[S_Omega_cap_I_k,n]}, and~\eqref{eq:master-eqs-genera-icl-hypergraph} follows from~\eqref{Bass_SI_hypernetworks_general}.
\end{proof}
In general, there are $2^M-1$ master equations in \eqref{eqs:master-eqs-general-hypergraph}, for all possible subsets $\Omega \subset \mathcal{M}$. Therefore, obtaining an explicit solution for general hypernetworks is not practical. In the following, we will obtain a considerably smaller reduced system of master equations for some special hypernetworks.

\section{Complete hypernetworks}
\label{complete_hypernetworks}
In~\cite{Niu2002, fibich2022compartmental}, it was shown that the expected adoption/infection level in the Bass/SI model on infinite complete homogeneous networks is the solution of the equations
\begin{equation}
	\label{2body_Bass_SI_infinite_complete}
	\frac{df}{dt}=(1-f)(p+qf),\qquad f(0)=I^0,
\end{equation}
and is given by $f_\mathrm{Bass}(t):=\frac{1-e^{-\left(p+q\right)t}}{1+\frac{q}{p}e^{-\left(p+q\right)t}}$ for the Bass model, and by $f_\mathrm{SI}(t):=\frac{ I^0}{e^{-q t} + (1 - e^{-q t}) I^0}$ for the SI~model.
In this section, we adopt a similar approach and compute explicitly, without making any approximation, the infinite-population limit of the Bass/SI model on complete 3-body hypernetworks, where
\begin{subequations}
\label{eq:p_j_q_j_complete-homog-hyper}
  	\begin{equation}
	I_j^0\equiv I^0, \qquad p_j\equiv p,\qquad
	q_{k_1,k_2,j}=\frac{q}{\binom{M-1}{2}}
	\mathds{1}_{k_1 \neq k_2,\,j \neq k_1,\,j \neq k_2}
	, \qquad \qquad j,k_1,k_2 \in \mathcal{M},
\end{equation}
$p>0$, $q>0$, and $I^0=0$ in the Bass model and $p=0$, $q>0$, and $0< I^0<1$ in the SI model.
The adoption rate of~$j$ is, see~\eqref{eq:lambda_j_3body(t)}, 
\begin{equation}
	\label{eq:lambda_j(t)-Bass-model-hypergraph-3-body-complete}
	\lambda_j^{\rm complete}(t)
	 = p+	\frac{q}{\binom{M-1}{2}} \sum_{k_1,k_2=1}^M \mathds{1}_{k_1 \neq k_2,\,j \neq k_1,\,j \neq k_2} \, X_{k_1}(t) X_{k_2}(t)
	 =  p+	\frac{q}{\binom{M-1}{2}}{N(t) \choose 2},
\end{equation}
\end{subequations}
where $N(t)=\sum_{j=1}^{M}X_j(t)$ is the number of adopters/infected in the population.
\begin{theorem}
\label{thm:f_complete_infty-hyper}
Let $f^{\rm complete}(t;M)$ denote the expected adoption /infection level in the  Bass/SI  model \textup{(\ref{Bass_SI_hypernetworks_general},\ref{eq:p_j_q_j_complete-homog-hyper})}
on complete 3-body hypernetworks. Let~$f^{\rm complete}_{\infty}:=\lim_{M \to \infty} f^{\rm complete}$. Then~$f^{\rm complete}_{\infty}$ is the solution of the equation
\begin{equation}
	\frac{df}{dt} = (1-f)(p+qf^2), \qquad f(0) = I^0.
	\label{eq:f_complete_infty-hyper}
\end{equation}
Furthermore, $f^{\mathrm{complete}}_{\infty}(t)$ is given by the explicit inverse formula
\begin{subequations}
\label{f_complete_explicit}
\begin{align}
	\label{f_complete_explicit_with_p}
	\left(p+q\right)t	=&\sqrt{\frac{q}{p}}\Biggl(
	\tan^{-1}\left(\sqrt{\frac{q}{p}}f^{\mathrm{complete}}_{\infty}\right)-\tan^{-1}\left(\sqrt{\frac{q}{p}}I^0\right)
	\Biggr)
	+\ln\left(\frac{1-I^0}{1-f^{\mathrm{complete}}_{\infty}}\right) \\
	& +\frac{1}{2}\ln\left(\frac{p+q\left(f^{\mathrm{complete}}_{\infty}\right)^{2}}{p+q(I^0)^{2}}\right), 
	\qquad \qquad \qquad \qquad \qquad \qquad \,\,\,   \mathrm{if} \,\,  p>0, \nonumber
\end{align}
and
\begin{equation}
	\label{f_complete_explicit_SI}
	q t = \ln \left(\frac{f^{\mathrm{complete}}_{\infty}}{1-f^{\mathrm{complete}}_{\infty}} \frac{1-I^0}{I^0}\right) + \frac{1}{I^0} - \frac{1}{f^{\mathrm{complete}}_{\infty}},\qquad \qquad \quad \mathrm{if} \,\,   p=0.
\end{equation}
\end{subequations}

\end{theorem}

\begin{proof}
	 Because of the symmetry of the hypernetwork, $[S_{\Omega}]$  only depends on the number of nodes in $\Omega$, and not on the specific choices of nodes in $\Omega$. Therefore, we can denote by
	\begin{equation}
		\label{eq:[S^n]-complete-hyper}
		[S^n]:= [S_{\Omega}  \mid  |\Omega|=n]
	\end{equation}
	the probability that all the nodes in any specific subset of $n$~nodes are nonadopters (susceptibles) at time~$t$.
	Substituting~\eqref{eq:p_j_q_j_complete-homog-hyper} and~\eqref{eq:[S^n]-complete-hyper} in the master equations~\eqref{eq:master-eqs-general-hypergraph} gives 
\begin{subequations}
	\label{eqs:master-homog-complete-hyper}
	\begin{align}
		\label{eq:master-homog-complete-hyper}
		\begin{split}
			\frac{d[S^n]}{dt} = & -n\left(p+q\frac{(M-n)(M-n-1)}{(M-1)(M-2)}\right)[S^n]
			\\ & +nq\frac{(M-n)(M-n-1)}{(M-1)(M-2)}\left(2[S^{n+1}]-[S^{n+2}]\right),
		\end{split} 
		\qquad \,\,\,\, n=1,\dots,M-2,
		\\	
		\label{eq:master-homog-complete-M-1-M-hyper}
		\frac{d[S^n]}{dt} =&-np [S^n],\qquad\qquad\qquad\qquad\qquad
		\qquad\qquad\qquad \qquad\,\,\, n=M-1,M.
	\end{align}

\end{subequations}	
Holding $n$ fixed and letting $M  \to \infty$,~\eqref{eqs:master-homog-complete-hyper} approaches 
\begin{equation}
	\label{eqs:master-homog-complete-hyper-M=infty}
	\frac{d[S^n]}{dt} = -n\left(p+q\right)[S^n]
	+nq\left(2[S^{n+1}]-[S^{n+2}]\right),
	\qquad
	[S^n](0)=(1-I^0)^n,\quad n = 1,2 \dots.
\end{equation}
The substitution $[S^n] = [S]^n$ reduces the infinite system~\eqref{eqs:master-homog-complete-hyper-M=infty} to the single ODE
\begin{equation*}
		\frac{d[S]}{dt} =  -(p+q)[S]+q\left(2[S]^{2}-[S]^{3}\right)=-[S]\left(p+q(1- 2[S]+[S]^{2})\right),\quad  [S](0) = 1-I^0.
		\label{S_complete_ODE}
\end{equation*}
 Substituting $f = 1-[S]$ gives~\eqref{eq:f_complete_infty-hyper}.
 Using partial fractions,
 $$
 \frac{1}{(1-f)(p+qf^2)} = \frac{1}{p+q} \left(q\frac{1+f}{p+qf^2}+\frac{1}{1-f} \right).
 $$
 Integrating and using $f(0)=I^0$ gives~\eqref{f_complete_explicit_with_p}. Taking the limit $p\rightarrow 0^+$ in~\eqref{f_complete_explicit_with_p} gives~\eqref{f_complete_explicit_SI}.
 \end{proof}
 
 In Figure~\ref{netHMF}A we compute numerically the expected adoption level in the Bass model on a complete 3-body hypernetwork with $M=5000$ nodes. The result is nearly indistinguishable from the explicit solution~\eqref{f_complete_explicit_with_p} on an infinite hypernetwork.	A similar agreement is observed in Figure~\ref{netHMF}B for the expected infection level in the SI model.
 
We note that, from Eq.~\eqref{f_complete_explicit} it follows that the time for half of the population to become adopters (infected) in the Bass/SI model on infinite complete hypernetworks is
\begin{equation*}
	\label{T_half_general}
	T_{\frac{1}{2}}=\frac{1}{p+q}\Biggl(
	\sqrt{\frac{q}{p}}\left(\tan^{-1}
	\biggl(
	\frac{1}{2}\sqrt{\frac{q}{p}}
	\biggr)
	-\tan^{-1}\left(\sqrt{\frac{q}{p}}I^0\right)\right)+\ln
	\biggl(
	\frac{1-I^0}{1-\frac{1}{2}}\sqrt{\frac{p+\frac{q}{4}}{p+q(I^0)^{2}}}
	\biggr)
	\Biggr).
\end{equation*}
For example, in the Bass model $I^0=0$, and so
\begin{equation*}
	T_{\frac{1}{2}}^{\mathrm{Bass}}=\frac{1}{p+q}\Biggl(
	\sqrt{\tilde{q}}\biggl(\tan^{-1}\Bigl(\frac{1}{2}\sqrt{\tilde{q}}\Bigr)\Biggr)+\ln\biggl(2\sqrt{1+\frac{1}{4}\tilde{q}}\biggr)
	\Biggr),\qquad \qquad \tilde{q}:=\frac{q}{p}.
\end{equation*}
In the SI model $p=0$, and so
\begin{equation*}
	T_{\frac{1}{2}}^{\mathrm{SI}}=\frac{1}{q}\Biggl(\ln \left(\frac{1-I^0}{I^0}\right)
	+\frac{1}{I^0}-2 \Biggr).
\end{equation*}

\begin{figure}[t]
	\centering
	\begin{subfigure}[a]{0.44\linewidth}
	\centering
	\begin{overpic}[width=1\linewidth]{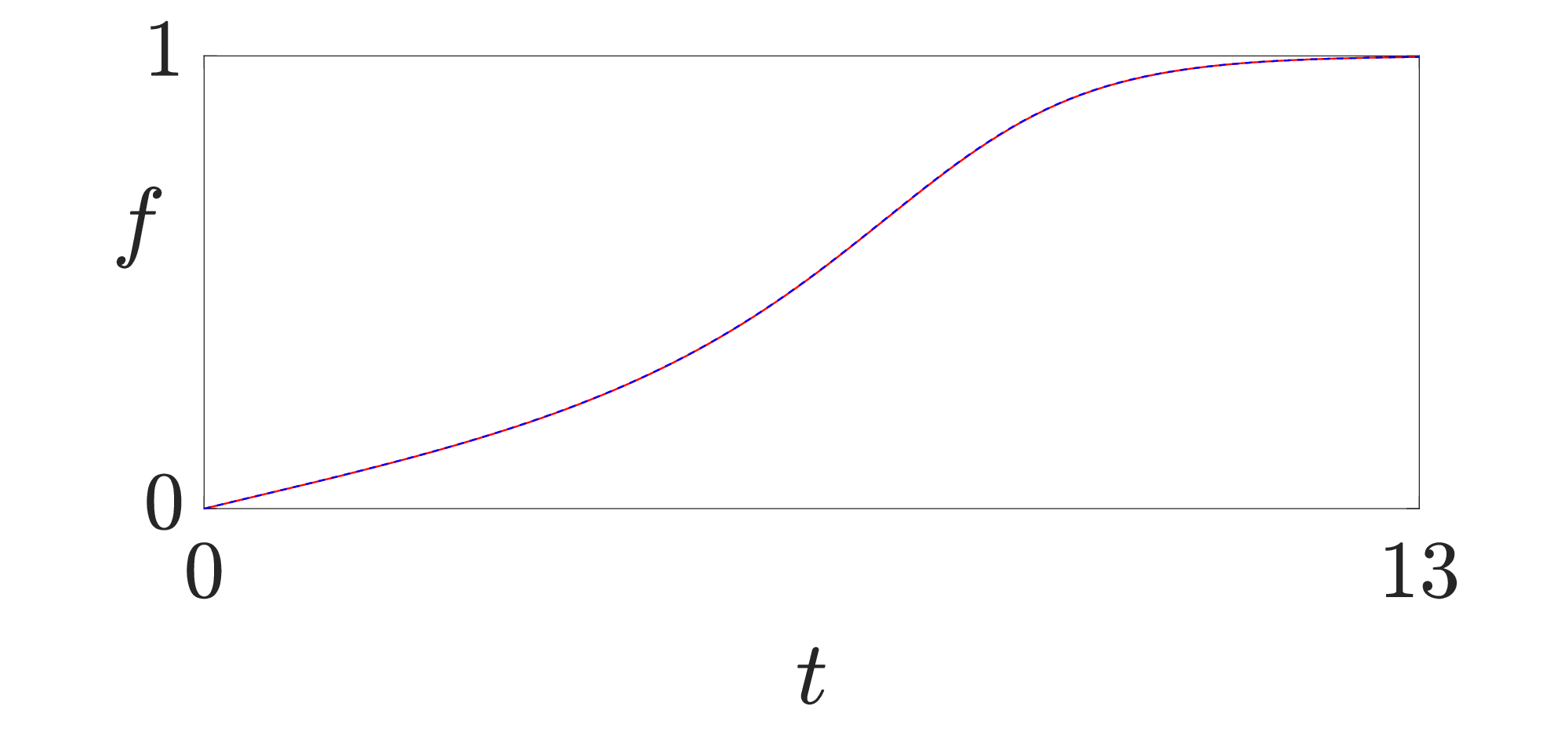}
	\put(15,38){(A)}
	\end{overpic}
	\end{subfigure}
	\hfill
	\begin{subfigure}[a]{0.44\linewidth}
	\centering
	\begin{overpic}[width=1\linewidth]{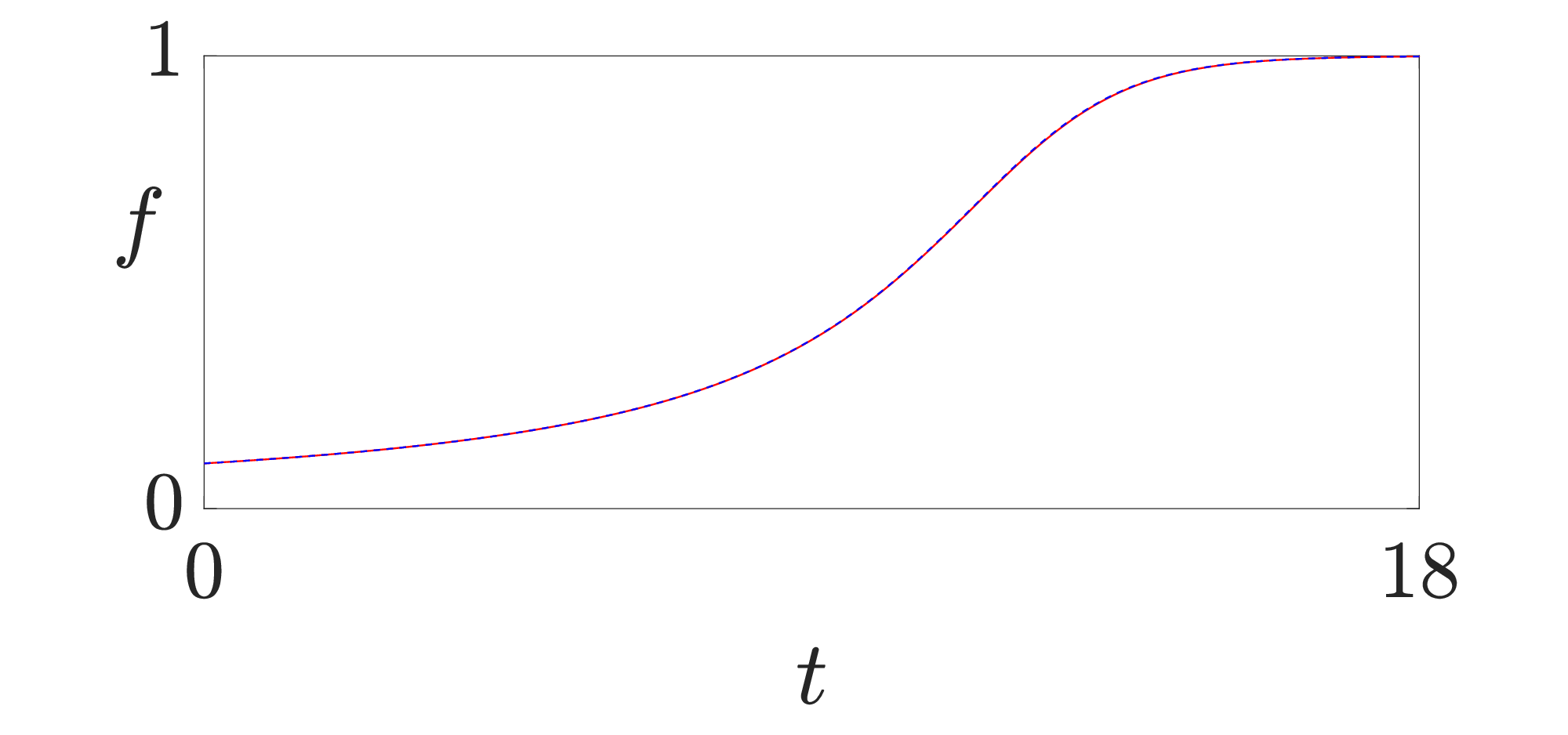}
	\put(15,38){(B)}
	\end{overpic}
	\end{subfigure}

	\caption {(A) The expected adoption level $f^{\text{complete}}(t)$ in the Bass model \textup{(\ref{Bass_SI_hypernetworks_general},\ref{eq:p_j_q_j_complete-homog-hyper})} on a complete 3-body hypernetwork with $M=5000$ nodes (solid red line) is nearly indistinguishable from the explicit expression~\eqref{f_complete_explicit} for $f^{\text{complete}}_{\infty}$ (blue dashed line). Here $p=0.05$, $q=1$, and $I^0=0$. (B)~Same as (A) for the SI model with $p=0$ and $I^0=0.1$.}
\label{netHMF}
\end{figure}

\section{Initial dynamics (Bass model)}
\label{initial_dynamics}

\mbox{\noindent
The expected adoption level in the Bass model on infinite complete hypernetworks satisfies, see~\eqref{eq:f_complete_infty-hyper},}
\begin{equation}
	\label{Bass_complete_hypernetwork}
	\frac{df}{dt} = (1-f)(p+qf^2), \qquad f(0) = 0,
\end{equation}
where $p,q>0$. Therefore, $\frac{df}{dt}$ {\em initially decreases} from $\frac{df}{dt}(0)=p$ to a local minimum, then increases to a global maximum, and finally decreases to zero (Fig~\ref{fCompletePhaseLine}A). This initial dynamics is qualitatively different from that on 2-body infinite complete networks, where
\begin{equation}
	\label{Bass_complete_network}
	\frac{df}{dt} = (1-f)(p+qf), \qquad f(0) = 0,
\end{equation}
see \eqref{2body_Bass_SI_infinite_complete}, and so $\frac{df}{dt}$ is an inverted parabola in $f$. In particular, if $q>p$, $\frac{df}{dt}$ increases from~$\frac{df}{dt}(0)=p$~to a global maximum and then decreases to zero (Fig~\ref{fCompletePhaseLine}B).

The initial decline of the adoption rate is not limited to complete hypernetworks. Indeed, it occurs for all 3-body hypernetworks \eqref{Set_notation}:

\begin{figure}[t]
	\centering
	\begin{subfigure}[b]{0.44\linewidth}
	\begin{overpic}[width=1\linewidth]{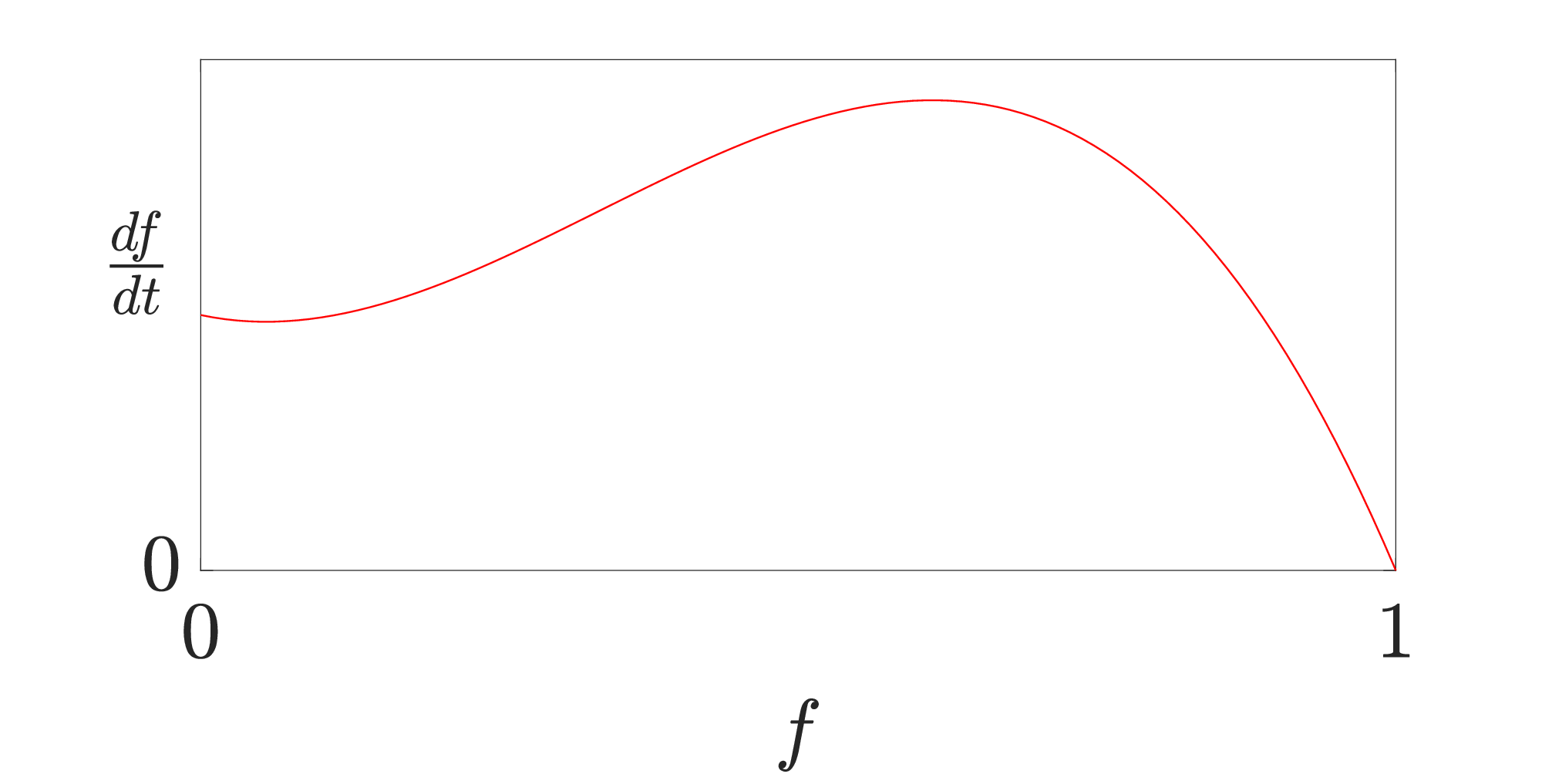}
		\centering
		\put(15,38){(A)}
	\end{overpic}
	\end{subfigure}
	\hfill
	\begin{subfigure}[b]{0.44\linewidth}
	\begin{overpic}[width=1\linewidth]{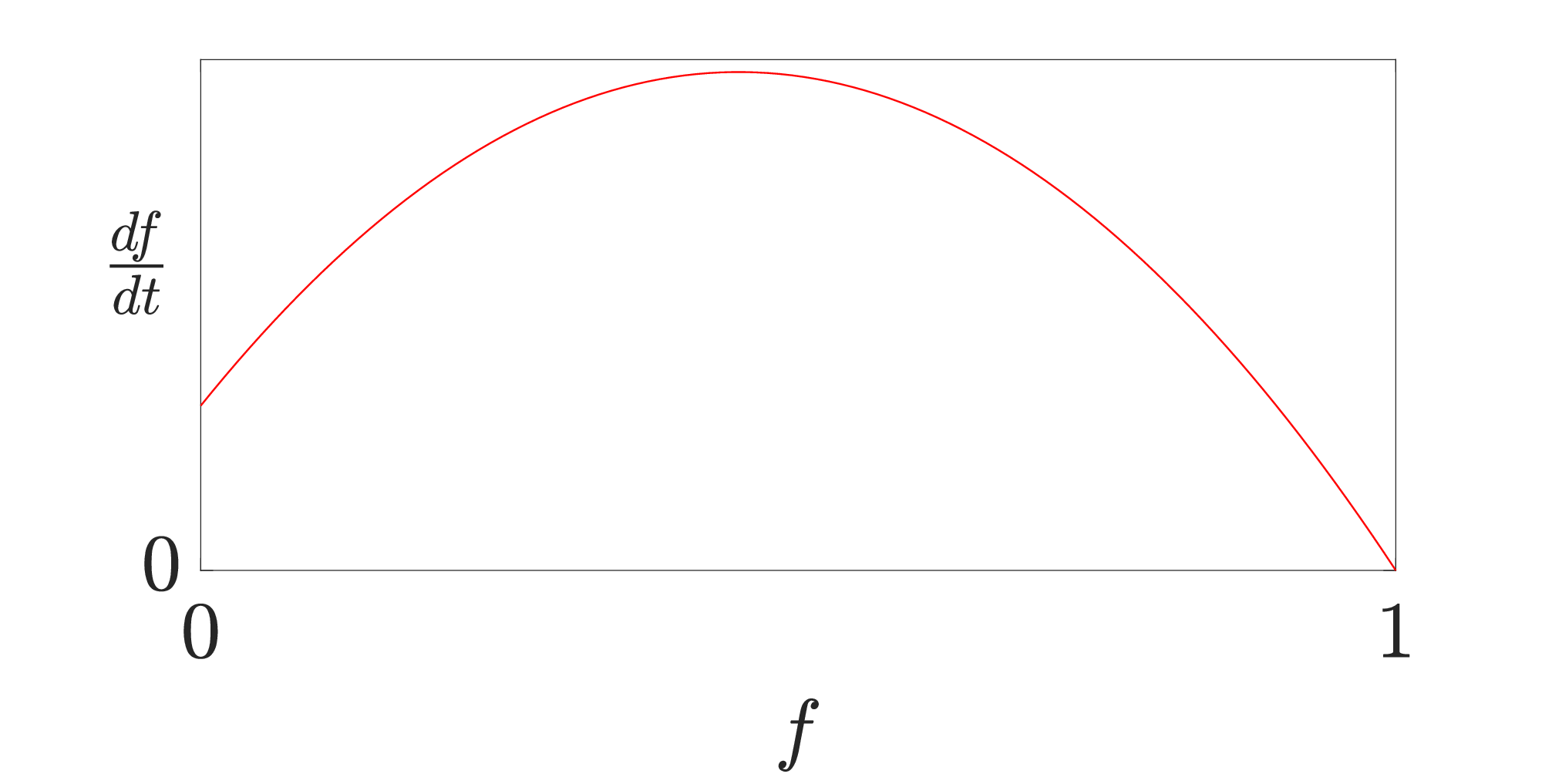}
	\centering
	\put(15,38){(B)}
	\end{overpic}
	\end{subfigure}

	\caption{ (A)~$\frac{df}{dt}$ as a function of $f$ for the Bass model \textup{(\ref{Bass_SI_hypernetworks_general},\ref{eq:p_j_q_j_complete-homog-hyper})} on an infinite complete 3-body hypernetwork. Here, $\frac{q}{p}=20$ and $I^0=0$. (B)~Same as (A) on an infinite complete 2-body network.   }
	\label{fCompletePhaseLine}
\end{figure}

\begin{theorem}
	Consider the Bass model~\eqref{Bass_SI_hypernetworks_general} on a 3-body hypernetwork. Then
	\begin{equation*}
		f'(0)>0,\qquad 
		f''(0)<0.
	\end{equation*}
\end{theorem}

\begin{proof}
	Substituting $\Omega = \{j\}$ in the master equations~\eqref{eq:master-eqs-general-hypergraph} gives
	\begin{equation}
	\label{dS_j_dt}
	\frac{d\left[S_{j}\right]}{dt}=-p_{j}\left[S_{j}\right]-\sum_{k_{1},k_{2}=1}^M q_{k_{1},k_{2},j}\Bigl(
	\left[S_{j}\right]-\left[S_{k_{1},j}\right]-\left[S_{k_{2},j}\right]+\left[S_{k_{1},k_{2},j}\right]
	\Bigr).
	\end{equation}
	Substituting $t=0$ in~\eqref{dS_j_dt} and using the initial conditions
	$
		\left[S_{j}\right](0)=\left[S_{k_{1},j}\right](0)=\left[S_{k_{2},j}\right](0)=\left[S_{k_{1},k_{2},j}\right](0)=1
	$
	 gives
	\begin{equation}
		\label{eq11}
		\frac{d\left[S_{j}\right]}{dt}(0)=-p_j.
	\end{equation}
	Similarly,
	\begin{equation}
		\label{eq12}
		\frac{d\left[S_{k_1,j}\right]}{dt}(0)=-p_{k_1}-p_j,\quad
		\frac{d\left[S_{k_2,j}\right]}{dt}(0)=-p_{k_2}-p_j,\quad
		\frac{d\left[S_{k_1,k_2,j}\right]}{dt}(0)=-p_{k_1}-p_{k_2}-p_j.
	\end{equation}
	Differentiating~\eqref{dS_j_dt}, substituting $t=0$ and using~\eqref{eq11} and~\eqref{eq12} yields
	\begin{align*}
		\label{eq-2}
		\frac{d^{2}\left[S_{j}\right]}{dt^{2}}(0)
		&=-p_{j}\frac{d\left[S_{j}\right]}{dt}(0)-\sum_{k_{1},k_{2}=1}^M q_{k_{1},k_{2},j}\left.\frac{d}{dt}\Bigl(
		\left[S_{j}\right]-\left[S_{k_{1},j}\right]-\left[S_{k_{2},j}\right]+\left[S_{k_{1},k_{2},j}\right]
		\Bigr)\right|_{t=0} \\
		&=p_{j}^{2}-\sum_{k_{1},k_{2}=1}^M q_{k_{1},k_{2},j}\left(-p_{j}+p_{k_{1}}+p_{j}+p_{k_{2}}+p_{j}-p_{k_{1}}-p_{k_{2}}-p_{j}\right)=p_{j}^{2},
	\end{align*}
Since $f=1-\frac{1}{M}\sum_{j=1}^M\left[S_{j}\right]$,
	$$\left.\frac{df}{dt}\right|_{t=0}=-\frac{1}{M}\sum_{j=1}^M \left.\frac{d}{dt}\left[S_{j}\right]\right|_{t=0}=\frac{1}{M}\sum_{j=1}^M p_{j}>0,$$
	and
	$$\left.\frac{d^{2}f}{dt^{2}}\right|_{t=0}=-\frac{1}{M}\sum_{j=1}^M\left.\frac{d^{2}}{dt^{2}}\left[S_{j}\right]\right|_{t=0}=-\frac{1}{M}\sum_{j=1}^M p_{j}^{2}<0.$$
\end{proof}

Since $f'(0)>0$ and $f''(0)<0$, the adoption rate decreases initially on all 3-body hypernetworks, regardless of the hypernetwork structure or of the ratio $\frac{q}{p}$. 

In order to understand the difference in the initial dynamics  in the Bass model between networks and hypernetworks, we note that as the adoption level $f$ increases, the adoption rate~$\frac{df}{dt}$ is influenced by 2 opposing mechanisms:
\begin{enumerate}
	\item The rate of internal influences increases.
	\item The rate of external influences decreases (since there are fewer non-adopters).
\end{enumerate}

The initial decrease of external adoptions is captured in~$f''(0)$. On 2-body networks, the initial increase of internal adoptions is also captured in~$f''(0)$. Therefore, the sign of $f''(0)$ depends on~$\frac{q}{p}$. On 3-body hypernetworks, however, a hyperedge becomes active only after two nodes adopt. Therefore, the initial increase of internal influence only enters to~$f'''(0)$. Hence,~$f'''(0)$ is always negative.

\section{Erd\H{o}s-Rényi hypernetworks}
\label{ER_hypernetworks}
In Erd\H{o}s-Rényi (ER) 3-body hypergraphs with $M$~nodes, for any three distinct nodes $k_1,k_2, k_3\in {\cal M}$, {\em  the hyperedge $\{k_1, k_2,k_3\}$ exists with probability~$\alpha$}, 
independently of all other hyperedges. Consider the Bass/SI model~\eqref{Bass_SI_hypernetworks_general} on ER hypernetworks, such that
\begin{subequations}
\label{eq:ER-network-dense}
\begin{equation}
	I_{k_1}^0=I^0, \qquad p_{k_1} \equiv p, \qquad 
	q_{k_1,k_2,k_3}=\frac{q}{\binom{M-1}{2}} \,e_{{k_1,k_2,k_3}},\qquad k_1,k_2,k_3\in{\cal M},
\end{equation} 
where ${\bf{E}}=\left(e_{{k_1,k_2,k_3}}\right)$ is the adjacency tensor, such that $e_{{k_1,k_2,k_3}}=1$ if there is a hyperedge connecting $\{k_1, k_2,k_3\}$, and $0$ otherwise.
The adoption rate of~$j$ is thus
	\begin{equation}
		\label{eq:lambda_j^ER}
		\lambda_{j}^{{\rm ER}}(t)=
		p+\frac{q}{\binom{M-1}{2}}\sum_{\{k_1,k_2\}\subset \mathcal{M}} e_{k_1,k_2,j}X_{k_1}(t)X_{k_2}(t).
	\end{equation}
\end{subequations}
\begin{lemma}
Let $f^{\rm ER}$ denote the expected adoption/infection level in the Bass/SI model \textup{(\ref{Bass_SI_hypernetworks_general}),\ref{eq:ER-network-dense})} on infinite ER 3-body hypernetworks. Then
\begin{equation}
	\label{f_ER_prediction}
f^{\rm ER}(t;p,q,\alpha,I^0)= f^{\rm complete}_\infty(t;p,\alpha q,I^0),
\end{equation}
where $f^{\rm complete}_\infty$ is the solution of~\eqref{eq:f_complete_infty-hyper}.
	\end{lemma}
\begin{proof}[Informal proof]
	When the  ER hypernetwork is large, we can apply the {\em mean-field} approximation, and replace $e_{k_1,k_2,j}$ with its expected value $\mathbb{E}\left(e_{k_1,k_2,j}\right)=\alpha$, i.e.,
	$$
	\sum_{\{k_1,k_2\}\subset \mathcal{M}} e_{k_1,k_2,j}X_{k_1}(t)X_{k_2}(t) \approx
	 \sum_{\{k_1,k_2\}\subset \mathcal{M}} \alpha \, X_{k_1}(t)X_{k_2}(t)=\alpha \binom{N(t)}{2}.
	$$
	Therefore, using~\eqref{eq:lambda_j(t)-Bass-model-hypergraph-3-body-complete} and~\eqref{eq:lambda_j^ER}, 
	$$
	\lambda_{j}^{{\rm ER}}(t;p,q) \approx \lambda_{j}^{{\rm complete}}(t;p,\alpha q). 
	$$
     Hence, the result follows from Theorem~\ref{thm:f_complete_infty-hyper}.
\end{proof}

In Figure~\ref{fig:er_graphs_figure} we compare the expected adoption/infection level $f^\mathrm{ER}$ in the Bass and SI models on ER hypernetworks with $M=1500$ nodes and the theoretical prediction~\eqref{f_ER_prediction}.
We let $q:=\frac{1}{\alpha}$, so that $f^\mathrm{complete}_{\infty}(t;p,\alpha q,I^0)$ remains unchanged as we vary $\alpha$. These calculations show that $f^\mathrm{ER}$ closely matches the theoretical prediction~\eqref{f_ER_prediction} for $\alpha = 0.5$, but not for $\alpha = 10^{-5}$. This is to be expected, since the mean-field approximation is derived for dense networks, and not for sparse ones. At $\alpha = 10^{-3}$ the hypernetwork is not dense, yet $f^\mathrm{ER}$ is still in good agreement with~\eqref{f_ER_prediction}. This can be attributed to the average hyperdegree $\langle k \rangle = \alpha \binom{M-1}{2}= 10^{-3} \binom{1499}{2} \approx~1100$, which is still large.

\begin{figure}[t]
	\centering
	\begin{overpic}[width=1\linewidth]{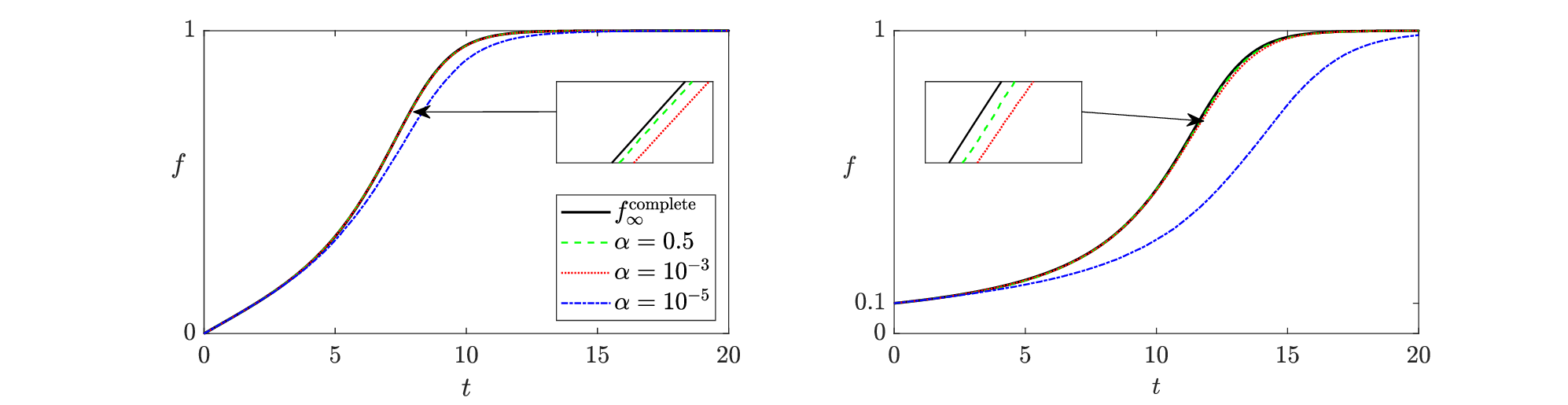}
		\centering
		\put(14,20.5){(A)}
		\put(59,20.5){(B)}
	\end{overpic}
	\caption{(A) The expected adoption level $f^{\text{ER}}$ in the Bass model \eqref{Bass_SI_hypernetworks_general} on ER hypernetworks with $\alpha = 0.5$ (dashed green line), $\alpha = 10^{-3}$ (red dashed-dotted line) and $\alpha = 10^{-5}$ (blue dotted line). The explicit solution $f^{\text{complete}}_\infty(t;p,\alpha q)$ (black solid line), see \eqref{f_complete_explicit}, is nearly identical to \(f^\text{ER}\) with $\alpha = 0.5$ and $\alpha = 10^{-3}$. Here, $M=1500$, $q=\frac{1}{\alpha}$, $p=0.05$, and $I^0=0$. (B)~Same as (A) for the SI model with $p=0$ and $I^0=0.1$.}
	\label{fig:er_graphs_figure}
\end{figure}

\section{Infinite hyperlines}
\label{infinite_hyperlines}
The expected adoption/infection level in the Bass/SI model on an infinite line satisfies~\cite{fibich2022compartmental,Gadi_Monotone_convergence}
\begin{equation*}
	\frac{df}{dt}=(1-f)\Bigl(p+q(1-e^{-pt})\Bigr),\qquad f(0)=I^0,
\end{equation*}
and is given by the explicit formula $f^\mathrm{1D}_\mathrm{Bass}(t;p,q):=1-e^{-(p+q)t+q\frac{1-e^{-pt}}{p}}$ for the Bass model, and by~$f^\mathrm{1D}_\mathrm{SI}\left(t;q,I^0\right):=1-\left(1-I^0\right)e^{-I^0qt}$ for the SI model. In this section, we derive the corresponding expressions for the Bass/SI model on the infinite homogeneous 3-body hyperline, where
\begin{subequations}
	\label{eq:p_j_q_j_circle-homog-hyper}
\begin{equation}
	I_j^0=I^0,\qquad p_j\equiv p,\qquad
	q_{k_1,k_2,j}=
	\begin{cases}
		q^{\rm L}, & \ {\rm if }\ j=k_1-1=k_2-2,\\
		q^{\rm R}, & \ {\rm if }\ j=k_1+1=k_2+2,\\
		0, &   \ {\rm otherwise },
	\end{cases}
	\qquad  \qquad k_1,k_2,j \in \mathbb{Z}.
\end{equation}
The adoption rate of~$j$ is, see~\eqref{eq:lambda_j_3body(t)}, 
\begin{equation}
	\label{eq:lambda_j(t)-Bass-model-hypergraph-hyperline}
	\begin{split}
		\lambda_j^{\rm 1D}(t)
		&= p+q^L X_{j-1}(t)  X_{j-2}(t) + q^R X_{j+1}(t)  X_{j+2}(t).
	\end{split}
\end{equation}
\end{subequations}
Note that when $q^{\rm L} \not=q^{\rm R}$, the hyperline is anisotropic. 
\begin{theorem}
	Let $f^{\rm 1D}$ denote the expected adoption /infection level in the Bass/SI model \textup{(\ref{Bass_SI_hypernetworks_general},\ref{eq:p_j_q_j_circle-homog-hyper})} on the infinite 3-body hyperline. Then $f^{\rm 1D}$ is the solution of
	\begin{equation}
		\label{f_1D_ODE}
	\frac{df^{\rm 1D}}{dt}  = \biggl(p+q \left(   1-(1-I^0)e^{-p t}\right)    ^2\biggr)\left(1-f^{\rm 1D}\right),
	\quad f^{\rm 1D}(0)=I^0,
	\end{equation}
	where $	q = q^{\rm L}+q^{\rm R}$, and is given explicitly by
	\begin{equation}
	\label{f_1D_explicit}
	f^{{\rm 1D}}=\begin{cases}
		1-\left[S^{0}\right]\exp\Biggl(-\left(p+q\right)t+\frac{q}{2p}\Bigl(1-e^{-pt}\Bigr)\Bigl(4-\left[S^{0}\right]-\left[S^{0}\right]e^{-pt}\Bigr)\left[S^{0}\right]\Biggr), & \mathrm{if}\,\, p>0,\\
		1-\left[S^{0}\right]e^{-\left(1-\left[S^{0}\right]\right)^{2}qt}, & \mathrm{if}\,\, p=0,
	\end{cases}
	\end{equation}
and $\left[S^{0}\right]=1-I^0$.
\end{theorem}

\begin{proof}
	Let $\Omega_n = \{m_1, \dots, m_n\}$ be a ``cluster'' of $n$~nodes, such that
	$m_1 < \dots < m_n$ and $m_{i+1}-m_i \in \{1,2\}$ for $i=1, \dots, n-1$. Note that if $k$ is a node such that $m_i<k<m_{i+1}$, then $k$ cannot affect the adoption rate of the nodes~in~ $\Omega_n$~(see~Figure~\ref{hyperline_example}~for~an~illustration). Therefore, $[S_{\Omega _n}]$  only depends on $n$, and not on the specific choices of nodes in $\Omega _n$.
	The master equation for $\Omega_n$ reads,~see~\eqref{eqs:master-eqs-general-hypergraph} and~\eqref{eq:p_j_q_j_circle-homog-hyper}
	\begin{align*}
		\begin{split}
			\frac{d[S_{\Omega_n}]}{dt}  = -(np+q^{\rm L}+q^{\rm R} ) [S_{{\Omega_n}}]
			&
			+q^{\rm L}
			\left([S_{\Omega_n,m_1-1} ]+[S_{\Omega_n,m_1-2 }]-[S_{\Omega_n,m_1-2,m_1-1}]\right) 
			\\
			&+q^{\rm R}
			\left([S_{\Omega_n, m_n+1}]+[S_{\Omega_n, m_n+2}]-[S_{\Omega_n, m_n+1,m_n+2}]\right),
		\end{split}
	\end{align*}
	subject to $\left[S_{\Omega_{n}}\right]\left(0\right)=\left(1-\left[I_{0}\right]\right)^n$. 
	Therefore, if we denote 
	$[S^n]:= [S_{{\Omega_n}}]$, then
	\begin{align}
		\label{S_n_line_ODE}
		\frac{d\left[ S^n \right]}{dt} = -(np+q ) [S^n]
		+q 		\left(2[S^{n+1} ]-[S^{n+2}]\right),\quad [S^n](0)=(1-I^0)^n, \quad n = \ 1,2, \dots, 
	\end{align}
We introduce the ansatz 
\begin{equation*}
	[S^n] = e^{-(n-1)pt}\left(1-I^0\right)^{n-1}[S^{\rm 1D}],\qquad n=1,2,...,
\end{equation*}
and substitute it in~\eqref{S_n_line_ODE} to reduce the infinite system~\eqref{S_n_line_ODE} to the single ODE
	\begin{equation}
		\label{S_1D_ODE}
		\frac{d\left[S^{\rm 1D}\right]}{dt}  = -\Biggl(p+q \Bigl(
		1-e^{-p t}\left(1-I^0\right)
		\Bigr)
		^2\Biggr)[S^{\rm 1D}],
		\quad [S^{\rm 1D}](0)=1-I^0.
	\end{equation}
Substituting $f^{\rm 1D}=1-[S^{\rm 1D}]$ in~\eqref{S_1D_ODE} gives~\eqref{f_1D_ODE}. Solving~\eqref{f_1D_ODE} gives~\eqref{f_1D_explicit} for $p>0$. Taking the limit $p\rightarrow 0^+$ of~\eqref{f_1D_explicit} yields~\eqref{f_1D_ODE} for $p=0$.
\end{proof}
\begin{figure}[t]
	\centering
		\begin{overpic}[width=1\linewidth]{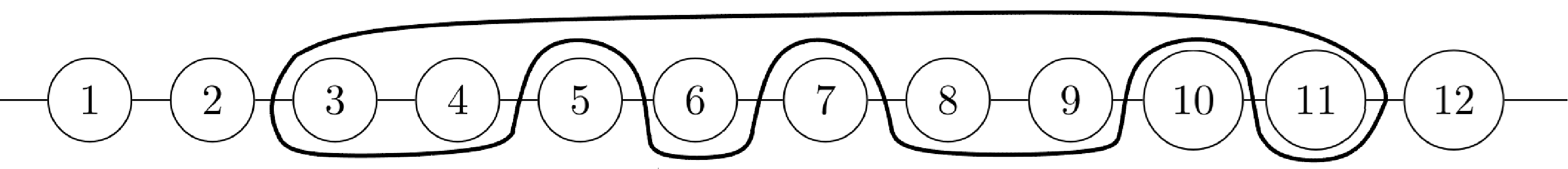}
		\centering
		\put(63,12){$\Omega _6$}
	\end{overpic}
	\caption{The set $\Omega_6 = \{3,4,6,8,9,11\}$ where $n=6$, $m_1 = 3$ and $m_6 = 11$. The nodes~$5,7$ and $10$ cannot affect $\left[S_{\Omega_6}\right]$.}
	\label{hyperline_example}
\end{figure}
Figure~\ref{ring} confirms that the expected adoption level in the Bass and SI models on the 3-body hyperline agrees with the theoretical prediction $f^{\rm 1D}$, see~\eqref{f_1D_explicit}.
\begin{figure}[H]
	\begin{subfigure}[a]{0.45\linewidth}
		\begin{overpic}[width=0.8\linewidth]{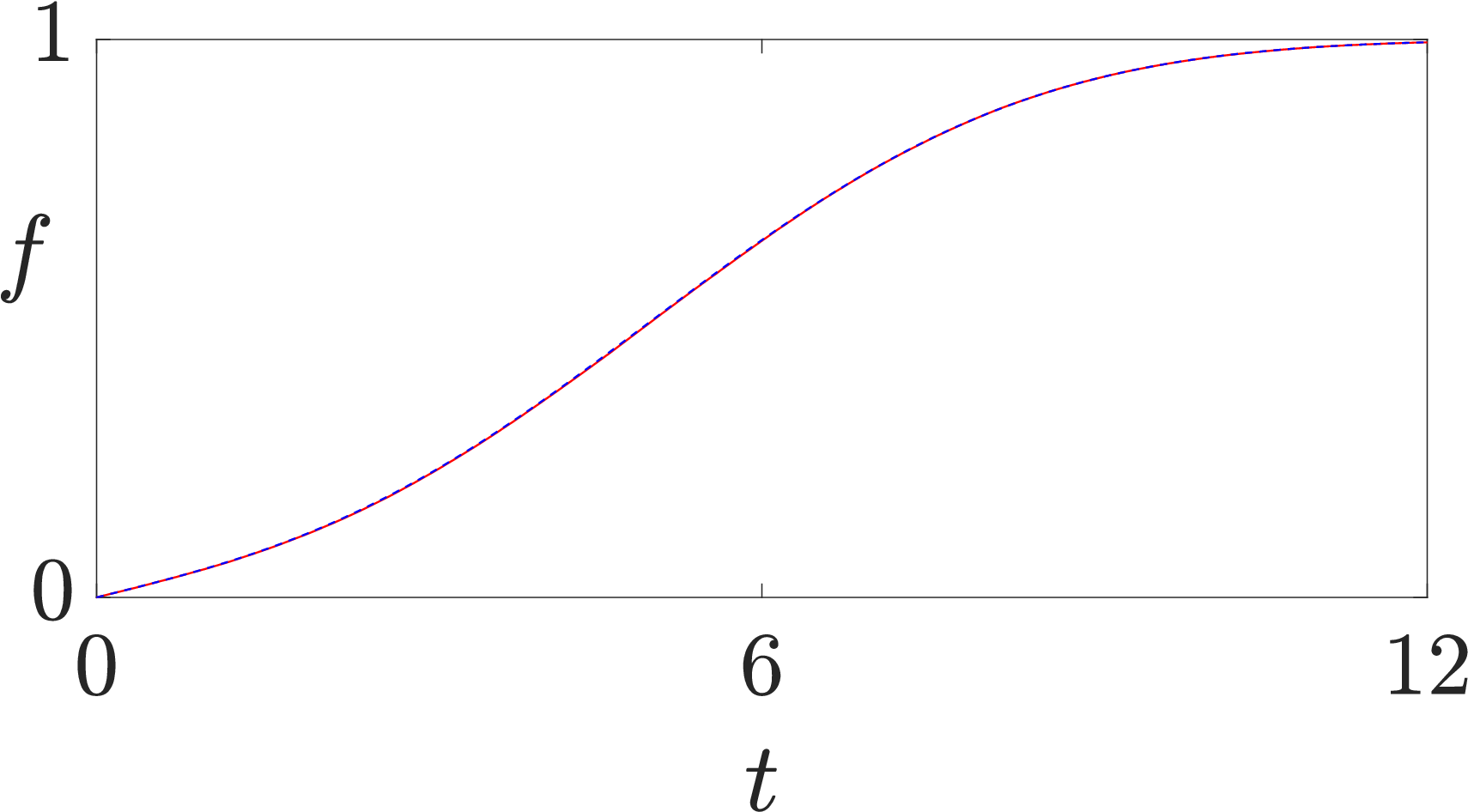}
			\centering
			\put(10,42){(A)}
		\end{overpic}
	\end{subfigure}
	\hfill
	\begin{subfigure}[a]{0.45\linewidth}
		\begin{overpic}[width=1\linewidth]{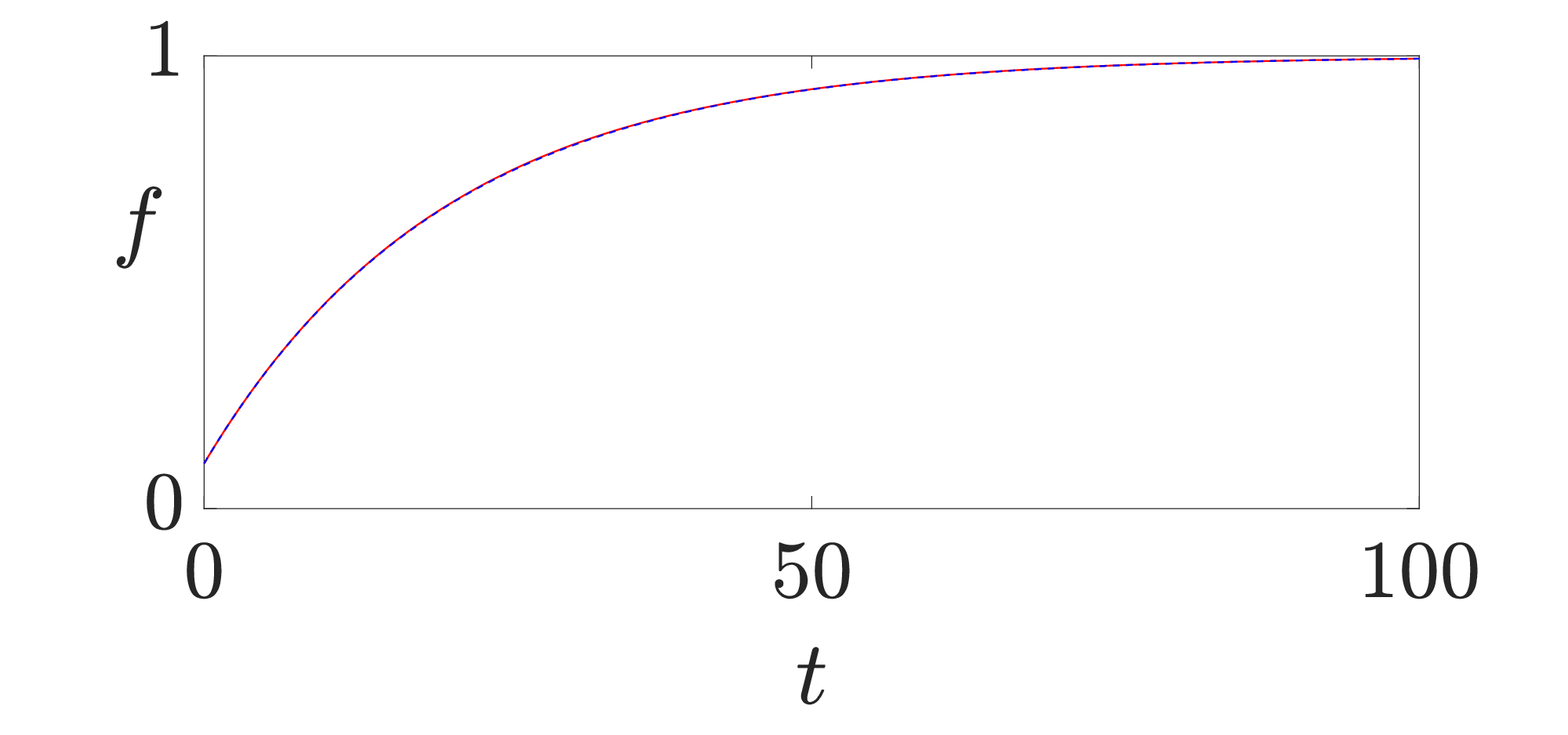}
			\centering
			\put(15,35){(B)}
		\end{overpic}
	\end{subfigure}
	\caption{The expected adoption level (solid red line) in the Bass/SI model \textup{(\ref{Bass_SI_hypernetworks_general},\ref{eq:p_j_q_j_circle-homog-hyper})} on a 3-body hyperline is indistinguishable from the theoretical prediction $f^{\rm 1D}$ (blue dashed line), see~\eqref{f_1D_explicit}. Here, $M=5000$, $q^L = 2$, $q^R = 3$, and $q=q^L+q^R$. (A)~Bass model:  $p=0.05$ and $I^0=0$. (B)~SI model: $p=0$ and $I^0=0.1$.}
	\label{ring}
\end{figure}

\section{Final Remarks}
\label{sec:final}

In this study we derived the master equations for the Bass and SI models on general hypernetworks with 3-body interactions. We then used these equations to obtain explicit exact solutions for several types of hypernetworks. In general, both the properties of explicit solutions, and the techniques used to derive them mimic those on two-body networks. In fact, the only qualitative difference between the two cases is the initial decline of the adoption rate (see section \ref{initial_dynamics}).

Extending our approach to hypernetworks with $N$-body interactions is straightforward. It is also natural to apply our approach to combinations of higher-order interactions (e.g. two- and three-body interactions). Finally, our methodology can be extended to other models on hypernetworks, such as SIS, SIR, and Bass-SIR, see \cite{Bass-SIR-analysis-17,Bass-SIR-model-16}.

\bibliographystyle{plain}


\end{document}